\pdfoutput=1
\documentclass[manuscript]{acmart}

\usepackage{amsfonts}
\usepackage{amssymb}
\usepackage{multirow}
\usepackage{amsthm}
\usepackage{graphicx}
\usepackage{subfigure} 
\usepackage{algorithm} 
\usepackage{algorithmic} 
\usepackage{bbm}

\theoremstyle{plain}

\usepackage{booktabs} % for professional tables
\newtheorem{theorem}{Theorem }[section]
\newtheorem{corollary}{Corollary }[theorem]
\newtheorem{lemma}[theorem]{Lemma }
\newtheorem{assumption}[theorem]{Assumption }

\AtBeginDocument{%
  \providecommand\BibTeX{{%
    \normalfont B\kern-0.5em{\scshape i\kern-0.25em b}\kern-0.8em\TeX}}}

\setcopyright{acmcopyright}
\copyrightyear{2020}
\acmYear{2020}
\setcopyright{acmlicensed}\acmConference[RecSys '20]{Fourteenth ACM Conference on Recommender Systems}{September 21--26, 2020}{Virtual Event, Brazil}
\acmBooktitle{Fourteenth ACM Conference on Recommender Systems (RecSys '20), September 21--26, 2020, Virtual Event, Brazil}
\acmPrice{15.00}
\acmDOI{10.1145/3383313.3412260}
\acmISBN{978-1-4503-7583-2/20/09}

\begin{document}

%%
%% The "title" command has an optional parameter,
%% allowing the author to define a "short title" to be used in page headers.
\title{Theoretical Modeling of the Iterative Properties of User Discovery in a Collaborative Filtering  Recommender System}

%%
%% The "author" command and its associated commands are used to define
%% the authors and their affiliations.
%% Of note is the shared affiliation of the first two authors, and the
%% "authornote" and "authornotemark" commands
%% used to denote shared contribution to the research.
\author{Sami Khenissi}
\email{sami.khenissi@louisville.edu}
\affiliation{%
  \institution{Knowledge Discovery \& Web Mining Lab,  University of Louisville}
  \streetaddress{2301 S 3rd St}
  \city{Louisville}
  \state{Kentucky}
  \postcode{40208}
}
\author{Mariem Boujelbene}
\email{mariem.boujelbene@louisville.edu}
\affiliation{%
  \institution{Knowledge Discovery \& Web Mining Lab, University of Louisville}
  \streetaddress{2301 S 3rd St}
  \city{Louisville}
  \state{Kentucky}
  \postcode{40208}
}
\author{Olfa Nasraoui}
\email{olfa.nasraoui@louisville.edu}
\affiliation{%
  \institution{Knowledge Discovery \& Web Mining Lab, University of Louisville}
  \streetaddress{2301 S 3rd St}
  \city{Louisville}
  \state{Kentucky}
  \postcode{40208}
}

%%
%% By default, the full list of authors will be used in the page
%% headers. Often, this list is too long, and will overlap
%% other information printed in the page headers. This command allows
%% the author to define a more concise list
%% of authors' names for this purpose.
\renewcommand{\shortauthors}{Khenissi, et al.}
\renewcommand{\shorttitle}{Theoretical Modeling of the Iterative Properties of User Discovery in Collaborative Filtering}

%%
%% The abstract is a short summary of the work to be presented in the
%% article.
\begin{abstract}
    The closed feedback loop in recommender systems is a common setting that can lead to different types of biases. Several studies have dealt with these biases by designing methods to mitigate their effect on the recommendations.  However, most existing studies do not consider the iterative behavior of the system where the closed feedback loop plays a crucial role in incorporating different biases into several parts of the recommendation steps.  

    We present a theoretical framework to model the asymptotic evolution of the different components of a recommender system operating within a feedback loop setting, and  derive theoretical bounds and convergence properties on quantifiable measures of the user discovery and blind spots. We also validate our theoretical findings empirically using a real-life dataset and empirically test the efficiency of a basic exploration strategy within our theoretical framework.  
    
    Our findings lay the theoretical basis for quantifying the effect of feedback loops and for designing Artificial Intelligence and machine learning algorithms that explicitly incorporate the iterative nature of feedback loops in the machine learning and recommendation process.
\end{abstract}
%%
%% The code below is generated by the tool at http://dl.acm.org/ccs.cfm.
%% Please copy and paste the code instead of the example below.
%%


%%
%% Keywords. The author(s) should pick words that accurately describe
%% the work being presented. Separate the keywords with commas.

%%
%% This command processes the author and affiliation and title
%% information and builds the first part of the formatted document.
\maketitle

\section{Introduction}
\label{introduction}

Concerns about the possible social impact of Artificial Intelligence and Machine Learning-driven recommender systems, particularly in the political and economical areas \cite{10.1145/3240323.3240370}, have given rise to disturbing questions concerning their consequences, including limiting human discovery, and as a result, human decision making and behavior. 
The social influence of recommender systems is particularly concerning because an increasing part of our online interactions consist of responses to the outputs of recommender systems on different platforms, such as social media, news, and entertainment websites. Our interactions with recommender systems naturally generate \textit{new data} which is then used to update or retrain the next iteration of recommender system models, thus effectively creating a closed feedback loop \cite{DBLP:conf/nips/SinhaGR16,nasraoui2016human} that affects the stream of information \textit{visible} to us, and consequently, our capacity to discover information online. Possible impacts of limited discovery include the creation of filter bubbles and polarization \cite{badami2017detecting,badami2018prcp,Badami}. Previous research has focused on dealing with this problem by aiming at increasing the diversity of the recommendation through several post-processing and preprocessing techniques \cite{Zhou4511,BarrazaUrbina2015XPLODIVAE} . Only a few have studied the \textit{iterative} behavior \cite{nasraoui2016human,kdir18}, and more specifically, the asymptotic behavior of the recommender system. 

In this paper, we start with presenting a theoretical model for the functioning of a generic Collaborative Filtering Recommender System. This model helps clarify our understanding of the different components of such a system, so that we are able to dissect the feedback loop. We then present a theoretical definition of the notions of \textit{user discovery} and \textit{blind spot} and study their iterative behavior by deriving asymptotic bounds that govern several metrics that can quantify these notions.
Because our proofs make use of certain assumptions, we also present an empirical study to assess the validity of these assumptions using a generic (common) recommender system model.
Finally, we validate our theoretical findings using an empirical study of the iterative closed loop interactions within several common recommender system algorithms using real and semi-synthetic data. We also evaluate the impact of common exploration strategies on the human discovery metrics and explore to what extent our theoretical bounds remain valid under these strategies. 

% this not really needed here since it is mentioned as future ideas- It can be used to replace that sentence in the Conclusions - better than mentioning delta-S and delta-B etc

Our theoretical and empirical findings suggest that future solutions for the feedback loop bias problem need to incorporate a dynamic component that can adapt to the iterative nature of the recommendation process. 

Our contributions can be summarized below:

\begin{itemize}
    \item We present a theoretical study of the evolution of a recommender system's feedback loop by looking at the asymptotic behavior of its components.
    
    \item We provide a theoretical bound on the human discovery resulting from interacting with a generic recommender system, under certain assumptions. 
    \item We verify our theoretical results, empirically \footnote{https://github.com/samikhenissi/TheoretUserModeling}, using a semi-synthetic dataset.
    \item We empirically study the limit of basic exploration approaches and their effectiveness in  increasing human discovery
\end{itemize}

\section{Related Work}

%%%%%%%%%%%%%%%%%%%%%%%%%%%%%%%%%
\label{Related Work}
%As the problem of bias in recommender systems is attracting a lot of public attention, many research have emerged trying to study and solve the problem. 

The term \textit{bias} in recommender systems has been used to represent a wide, but often related, variety of biases \cite{10.1145/3209581}. These include popularity bias, diversity bias, exposure bias, display bias, iterative bias, etc.  We will therefore summarize the bias types that we consider to be the most related to our work\footnote{We do not consider certain types of biases such as display bias, user specific bias, etc, which we leave  for future work.}. 

\subsection{Diversity Bias}
Diversity bias studies the variety between items within a recommendation list. This is probably the earliest type of bias studied, as it directly relates to the overfitting problem in machine learning \cite{KUNAVER2017154}. 
Some of the prior work \cite{Zhou4511} focused on the diversity of recommendations by incorporating different techniques to diversify the final recommendations. These techniques, including Determinantal Point Process  \cite{ddp}, regularization techniques \cite{IJCAI136511}, optimization techniques \cite{Cheng2017}, and so on, focus on the \textit{model learning step} of the recommender system algorithm. In fact, by designing a method that effectively \textit{learns} to promote the recommendation of diverse items, we can alleviate the problem of diversity bias.

Other methods have focused on post-processing methods by designing a new recommendation strategy based on exploration techniques, such as Multi-Armed-Bandits algorithms \cite{mab} or other post-filtering methods using personalized recommendation selections \cite{BarrazaUrbina2015XPLODIVAE}.

Diversity bias causes several known issues in recommender systems and information retrieval such as filter bubbles and polarization \cite{badami2018prcp,doi:10.1080/1369118X.2016.1271900}. 
We argue that diversity bias itself can be a direct consequence from the iterative behavior or the closed feedback loop in recommender systems. 

\subsection{Popularity Bias}

Another type of studied bias is popularity bias \cite{10.1145/2043932.2043957}. It often translates to studying the long tail items. In fact, popularity bias causes some unfairness between items, since popular items get increasingly promoted compared to other items \cite{unfair}.

In order to mitigate this bias, several techniques have emerged, ranging from regularization \cite{10.1145/3109859.3109912,8696023} to re-ranking and post-processing techniques \cite{Abdollahpouri2019ManagingPB,Jannach2015}. These methods aim at promoting the long tail items in order to provide a fair exposure to these items and reduce unfairness.

Even though it is often distinct from iterated bias (Section \ref{iteratedbias}), popularity bias is largely affected by the behavior of the closed feedback loop. In fact, for a collaborative filtering algorithm where no external data is included in the training and recommendation process, popular items get recommended often because their high number of ratings helps the algorithm learn to group them more accurately. This bias phenomenon can be considered to be related to the cold start problem \cite{10.1145/3109859.3109912}. Popularity bias is exacerbated when the user's feedback is used in the next iteration of training the model.

\subsection{Exposure Bias}

We can think of exposure bias as the iterative bias (Section \ref{iteratedbias}) studied in a \textit{single} fixed iteration. In this setting, we study how the user was exposed to the items before providing the ratings and then how this will affect the next predictions. This is equivalent to studying the Missing Not At Random (MNAR) problem \cite{10.1145/3240323.3240355}. In fact, by studying the distribution of the missing data, we can infer the effect of the bias on the predictions and/or the training. 

Within this scope, previous studies have tried to estimate the probability that the user is exposed to an item before rating it, and then use this estimate to alter the training algorithm by designing better estimators of the performance of the algorithm \cite{DBLP:journals/corr/SchnabelSSCJ16,NIPS2019_9628, pmlr-v97-wang19n, 10.1145/3209978.3209986} or using regularization techniques \cite{khenissi2020modeling}.

Exposure bias is directly related to the closed feedback loop in a recommendation system. In our work, we are interested in the temporal evolution of the exposed set and the missing data (i.e the unseen items).

\subsection{Iterative or Closed Feedback Loop Bias}
\label{iteratedbias}
\begin{figure}[ht]
\vskip 0.2in
\begin{center}
\centerline{\includegraphics[width=0.5\linewidth]{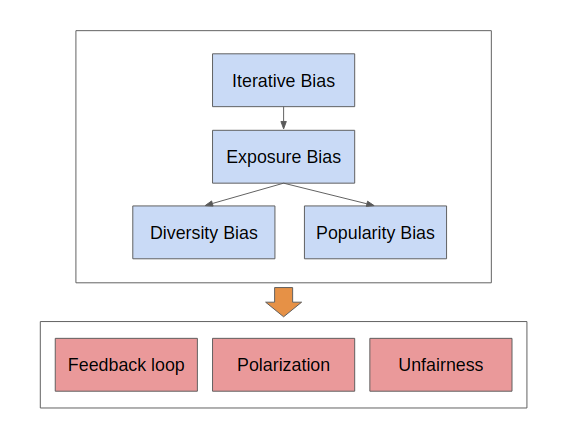}}
\caption{Hierarchical segmentation of bias types. Iterative bias is the most general type of bias as it is caused by the iterative interaction process of the recommender system with the user. This behavior is what causes several other bias types such as exposure bias and diversity bias. It also leads to popular phenomena such as Filter bubbles and Polarization}
\label{fig:hierar}
\end{center}
\vskip -0.2in
\end{figure}

If we classify the bias types into a hierarchical scheme, as shown in \textbf{Figure \ref{fig:hierar}},iterative or closed feedback loop bias would subsume all the aforementioned types of bias. Depending on the initial state of the system, the closed feedback loop incorporates the user selection bias, which is itself affected by the previously seen items, suggested by the previous recommendations \cite{10.1145/3308560.3317303}. Hence, understanding this iterative behaviour is important in order to control its effects. \cite{Jiang_2019} provided a theoretical and empirical study of how the iterative process can affect the user preferences, proving that if a user is trapped in a filter bubble with no strategy to escape, the recommendation quality will decrease. \cite{10.5555/3327757.3327773} modeled the feedback loop through dynamically modeling the dynamic item missingness of the ratings by combining Markov Models and Matrix Factorization in order to improve the quality of the recommendations. \cite{EPPERLEIN2019116} and \cite{nasraoui2016human} also used Markov models to model the closed loop.  Other studies \cite{10.1145/3287560.3287583,khenissi2020modeling} tried to empirically model the iterative bias by using simulations that study the effect of the algorithms on various diversity metrics. Furthermore, \cite{DBLP:conf/nips/SinhaGR16} tried to deconvolve the feedback loop to understand how it affects the final ratings 

Although these studies consider the temporal dependency in recommender systems, a simplified theoretical modeling of how the system evolves has been missing. Our work falls under this category, where we consider an \textit{iterative} collaborative filtering recommender system environment and further formulate assumptions under which the system evolves to a static state characterized by limited human discovery.

%%%%%%%%%%%%%%%%%%%%%%%%%%%%%%%%

\section{Theoretical results }
\subsection{Notation}
\label{Notation}

\begin{table*}[ht]
\caption{Summary of the notation used in the paper}

\label{tab:tabNotation}

\resizebox{\linewidth}{!}{%

\begin{tabular}{ll|ll|ll}

Symbol & Meaning                                  & Symbol  & Meaning                             & Symbol       & Meaning                                   \\
I      & Set of all items                          & G       & \ Set of item groups         & $B_t$        & Blind spot of the user up to iteration t       \\
U      & Set of Users                             & $Rec_t$ & recommended groups of items at iteration t    & $\Delta S_t$ & User discovery                            \\
$f_t$  & Recommender selection function           & Rel     & Set of relevant groups of items               & $\Delta B_t$ & Evolution of the blind spot               \\
t      & iteration                                & $S_t$   & Seen item groups up to iteration t & $| . |$      & Cardinality                               \\
$g$   & a group of items ($g\in G$)                           & $\mathbbm{1}_u$      & Indicator function of condition u   & $P_f$        & Probability that an item from a given group will be recommended  \\
$h_p$  & Ranking function related to measure $P_f$ & O       & Ordering score provided by P        & M            & Mapping function that maps items to groups            
\end{tabular}
}
\end{table*}
We start by defining the notation we will use throughout the paper that we also summarize in Table \ref{tab:tabNotation}. A recommender system in a collaborative filtering setting can be fully defined by the quadruple $(I,U,f_t,t)$ where:
\begin{itemize}
    \item $I$ represents the set of all available items. 
    \item $U$ is the set of users. \footnote{Users in this paper can refer to groups or neighborhoods of users (users that share the same characteristics based on a given distance measure).}
    \item $f_t: G \rightarrow I$  is the selection function of the recommender system which selects a fixed number of items to recommend to a given user at iteration $t$.
    \item $t$ is the iteration  number in the recommendation process as we are interested in the sequential dependency of the system.
\end{itemize}

In this paper, we are interested in studying the user discovery behaviour and therefore we study the interaction of "groups" of items, as was done by \cite{Jiang_2019}, rather than "single" elements. A group represents a neighborhood of similar items based on some criteria such as genre, type, or nature of items.
We therefore define $G$ as the set of item groups. This will be useful later to define the user discovery capacity and the blind spot. In fact, items that are from the same group of items (share the same type, genre or characteristics, etc) as items that have been seen before may not be considered to be significantly new and may be considered to not contribute to increase human discovery, as much as items from totally unseen groups. We also define $M$: $I \rightarrow G$, a mapping function that maps the items to a set of item groups.

We further define $Rec_t$ as the set of  item groups that are recommended to a user at iteration $t$ and $Rel$ as the set of  item groups that are relevant to a given user.

We also define the set $S_t$ which represents the set of item groups that have been seen up to iteration $t$. As this set  influences the \textit{discovery capacity} of the user, we are interested in its evolution throughout the iterations. We illustrate the process of the iterative recommendation setting studied in \textbf{Figure \ref{fig:process}}. 

\begin{figure}[ht]
\begin{center}
\centerline{\includegraphics[width=0.8\linewidth]{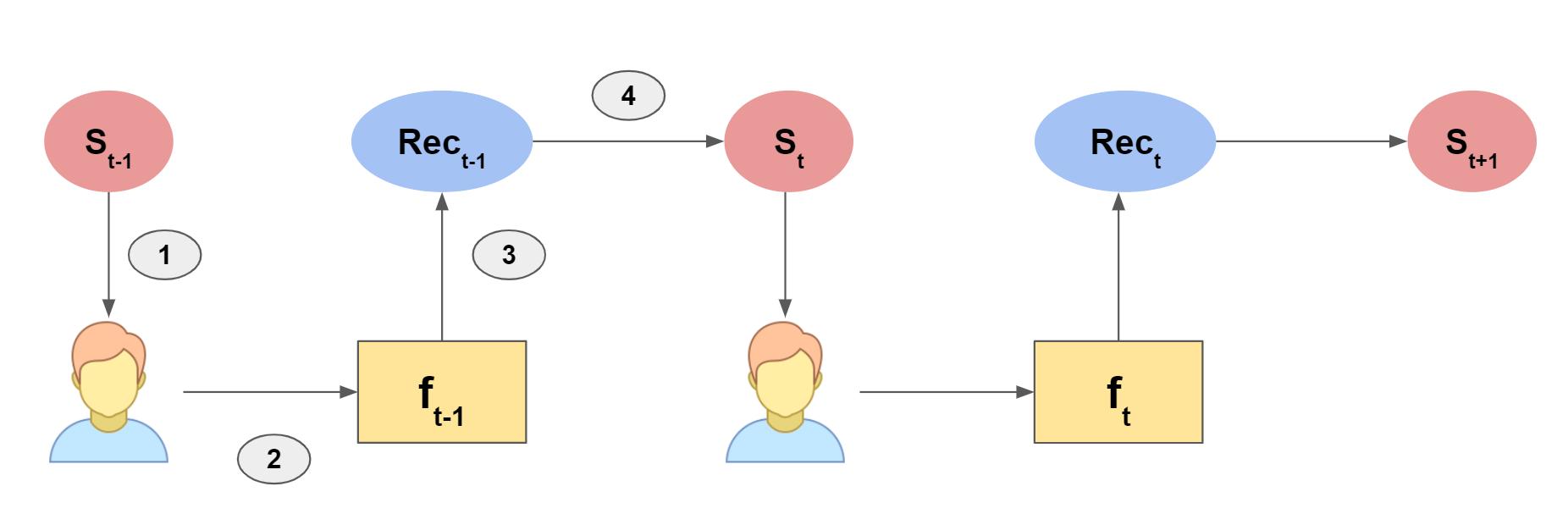}}
\caption{Illustrative example for an iterative recommendation process. \textbf{Step 1} represents the event of the user seeing some item groups $S_{t-1}$; then in \textbf{Step 2}, the user provides ratings to the seen items. In \textbf{Step 3}, the ratings are used by the recommender system function $f_{t-1}$ in order to provide recommendations represented by $Rec_{t-1}$. These recommendations are finally added in \textbf{Step 4} to the previously seen item sets to form $S_t$. The process continues for several iterations using the same steps.}
\label{fig:process}
\end{center}
\vskip -0.2in
\end{figure}

Finally, we define $B_t = S_t^c \cap Rel$ as the \textit{blind spot} of the user at iteration $t$. The blind spot in a recommender system defines the item groups that have not been seen by the user, \textit{although they are relevant}.

In practice, we want to increase the user's discovery ability (by increasing the size of $S_t$), but without recommending irrelevant groups. This goal translates to decreasing the size of the blind spot $B_t$.

We are interested in studying the evolution of $S_t$ and $B_t$ throughout increasing iterations. For this reason, we define:
\begin{itemize}
    \item $\Delta S_t = |S_t| - |S_{t-1}|$: The increase in the number of new groups of items between consecutive iterations.
    \item $\Delta B_t =  |B_{t-1}| - |B_t| $: The decrease in the number of unseen and relevant groups of items between consecutive iterations.
    \item $\sum_t \Delta S_t$ is the user's \textit{discovery capacity} and $\Delta_n S = \frac{1}{n} \sum_{t=0}^{n} \Delta S_t$ is the human average discovery.
\end{itemize}

In the next sections, we study the asymptotic behavior of $\Delta_n S$ and $\Delta_n B$ and show that under certain assumptions, $\Delta_n S$ decreases exponentially with the iterations. 

\subsection{Assumptions}
\label{Assumption}

We make three assumptions about the nature of the iterative recommendation process, that will allow us to study the interaction between the user and the recommender system, while being isolated from any extraneous factors. The assumptions will be later tested empirically in Sec 4.1 and further challenged by relaxing the ranking requirement by adding random exploration.

\begin{assumption}[Vacuum Assumption]
\label{A1}

Let $(I,U,f_t,t)$ be a recommender system, $S_t$ is the set of seen item groups and $Rec_t$ is the set of newly recommended item groups at iteration $t$. We assume that:
\begin{equation}
    S_{t+1} = S_t \cup Rec_t
\end{equation}
\end{assumption}

This assumption is necessary to allow us to study the impact of the recommender system in isolation of any other external factors. Although it would be interesting to study the impact of extraneous factors that could lead to the user seeing an item without interacting with the seen items, we believe that this would be premature at this stage. These additional factors could include display bias and user discovery occurring outside the recommender system itself.

\begin{assumption}[Perfect Feedback Assumption]
\label{A2}
Let $(I,U,f_t,t)$ be a recommender system and let $M$ be the item-to-group mapping function, $S_{t-1}$ be the set of seen groups of items, and $Rec_t$ be the set of newly recommended groups of items at iteration $t$. We assume that:
\begin{equation}
    Rec_t = M \circ f_t(S_{t-1})
\end{equation}
\end{assumption}

This assumption simply states that the recommender system provides the next set of recommendations using the information provided by $S_t$. This is the case for a collaborative filtering recommender system that relies on feedback in the form of ratings from the user on the seen item groups. This assumption also eliminates the possibility of the user seeing an item without rating it.

\begin{assumption}[Ranking Assumption]
\label{A3}
Let $g_i , g_j \in G$ be two groups of items and let $d:G \times G \rightarrow \mathbb{R}$ be a distance measure between two sets in $G$. If $d(g_i,S_{t-1}) \leq d(g_j,S_{t-1})$, then we have $P_f(g_j \in Rel) \leq P_f(g_i \in Rel)$  
\end{assumption}

What this assumption states is that a group of items has a higher probability to be recommended if it has been seen before.
In other words, the recommender system maximizes item relevance by ranking the items closest to those that have already been seen, and relevant groups of items, higher than other items. 

This is probably the strongest assumption we make but it has strong practical justifications. The assumption suggests that the recommendation strategy is purely exploitation-based and excludes strategies based on exploration such as Multi-Armed-Bandits methods. Most collaborative filtering recommender systems such as Matrix Factorization and Nearest Neighbors Methods, which are some of the most popular methods, satisfy this assumption, as they are trying to minimize a distance, embedded in the objective function, between the unseen items and the seen items. In Section \ref{Empval} we will present an empirical validation of this assumption. Then in Section \ref{explorationsection} we will challenge this assumption by adding an exploration component.

\subsection{Asymptotic behavior of human discovery}
It should not come as a surprise that a recommender system will eventually limit the variety of new items provided to the user. Our aim however is to provide a  \textit{theoretical framework} to estimate the speed of convergence of this discovery limitation and hopefully lay the theoretical ground for designing future machine learning techniques that strive to optimize the discovery capacity in addition to other measures of performance.

Our main results in this section will be Lemma \ref{theorem1} and Theorem \ref{theorem2}. Lemma \ref{theorem1} defines the mathematical nature of the \textit{Seen} item space and hence provides tools to explore and control the recommendation strategies defined in that space. Theorem \ref{theorem2} provides a convergence argument along with a theoretical bound that gives insight about the convergence speed of the user discovery capacity. We start by laying the groundwork in the following Lemma.
\begin{lemma}
Let $f$ be the generic selection function of a recommender system,  $S_t^c$ be the complement of $S_t$, and $g_j \in G$ be a group of items. if $f$ satisfies \textbf{Assumption \ref{A3}} and $|S_t|>|Rec_t|$ then 
\begin{equation}
    E(|\{g_j \in S_t^c \cap Rec_t\}|) = 0
\end{equation}
\begin{proof}
We first calculate the expected value of the number of recommended item groups. Since we have a fixed set of recommendations, we have:
\begin{equation*}
    E(|\{g_j \in G \cap Rec_t\}|) = |Rec|,
\end{equation*}
where $|Rec| = |Rec_t|$ $ \forall t.$ Then
\begin{equation*}
    E(|\{g_j \in (S_t \cup S_t^c) \cap Rec_t\}|) = |Rec|
\end{equation*}
\begin{equation*}
    E(|\{g_j \in (S_t \cap Rec_t) \cup  (S_t^c \cap Rec_t) \}|) = |Rec|
\end{equation*}
By the linearity of expectation, we have:
\begin{equation*}
    E(|\{g_j \in (S_t \cap Rec_t) \}|) +   E(|\{g_j \in S_t^c \cap Rec_t) \}|) = |Rec|
\end{equation*}
Next, we will find an expression for $P(g_j \in Rec_t) $. 

Let us denote the random variable $O = P_f(g_j \in Rel)$ that represents the score (for example predicted normalized rating) of each item group as provided by $f$, then we can define a ranking function $h_P: G \rightarrow \{0,1..|G|\}$ that maps $g_j$ to its rank among all groups of items. 

Considering the ordered statistics $\{ O_{(1)},O_{(2)},....,O_{(|G|)}  \}$, we get $h_P(g_j) = q $, where $ O_{(q)} =  P(g_j \in Rel) $.

Hence, $P(g_j \in Rec_t) = 1 $ if $h(g_j)  \leq |Rec_t| $ and $P(g_j \in Rec_t) = 0 $ otherwise; meaning that an item group is in the recommendation list if its rank is less than $|Rec_t|$. This is based on a purely exploitation strategy that only considers the rank of the relevance estimate of the item.

Therefore, we can affirm that  $P(g_j \in Rec_t) = \mathbbm{1}_{(h(g_j)  \leq |Rec_t|)}$. Thus

 \begin{equation*}
   E(|\{g_j \in S_t \cap Rec_t) \}|)  = \sum_{j=0}^{|S_t|} P(g_j \in Rec_t | g_j \in S_t)  
\end{equation*}

 \begin{equation*}
    E(|\{g_j \in S_t^c \cap Rec_t) \}|) = \sum_{i=0}^{|S_t^c|} P(g_i \in Rec_t | g_i \in S_t^c)  
\end{equation*}

 \begin{equation*}
   |Rec_t| = \sum_{j=0}^{|S_t|} \mathbbm{1}_{(h(g_j | g_j \in S)  \leq |Rec_t|)}   +\sum_{i=0}^{|S_t^c|} \mathbbm{1}_{(h(g_i | g_i \in S_t^c)  \leq |Rec_t|)}
\end{equation*}

and according to Assumption \ref{A3}, we have 
\begin{equation*}
 h(g_j | g_j \in S_t) \leq    h(g_j | g_j \in S_t^c).
 \end{equation*}

Thus, we have

\begin{equation*}
 max(h(g_j | g_j \in S_t)) \leq    min(h(g_j | g_j \in S_t^c)) \:  \forall g_j \in G
 \end{equation*}

 Assuming that $|Rec_t| \leq |S_t|$, we get $max(h(g_j | g_j \in S_t)) > |Rec_t|$, then $\mathbbm{1}_{(h(g_i | g_i \in S_t^c)  \leq |Rec_t|)} = 0 \:  \forall g_i \in S^c$. 
 
 Finally, we conclude that 
 
 \begin{equation*}
    \sum_{j=0}^{|S_t|} \mathbbm{1}_{(h(g_j | g_j \in S)  \leq |Rec_t|)} = |Rec_t|.
\end{equation*}
This means that
\begin{equation*}
    E(|\{g_j \in (S_t \cap Rec_t) \}|)  = |Rec_t|,
\end{equation*}
and finally
\begin{equation*}
    E(|\{g_j \in S_t^c \cap Rec_t) \}|) = 0.
\end{equation*}
\end{proof}
\end{lemma}
What this lemma is suggesting is that a ranking algorithm satisfying \textbf{Assumption \ref{A3}} will find an element to recommend from a group that has been seen in order to increase the relevance of the recommendations. We stress again that we are dealing with \textit{groups} of items. This means that the user may see a new item but this item will be from an already seen group. Also note the important role of the recommendation list length, as a long recommendation list may allow for further discovery; but in practice, recommendation lists have a limited length.

This lemma lays the groundwork for the next result, as it allows us to quantify the nature of the Seen Groups Space.

\begin{lemma}
\label{theorem1}
By defining the measurable space $(\Omega ,G)$, where $\Omega$ is the sample space of item groups, and if $E(|\{g_j \in S_t^c \cap Rec_t\}|) = 0$ ;
then $S_t$ is a filtration on $\Omega$ and the random process $|S_t|$ is a martingale defined in the filtered probability space $(\Omega,S,\mathbb{P})$,  where $\mathbb{P}$ is a probability measure in $(\Omega ,G)$.

\end{lemma}

\begin{proof}
the proof that $S_t$ is a filtration is straightforward from the definition since $\forall k<l$, we have $S_k \subset S_l$.

To prove that $|S_t|$ is a martingale, we need to first see that $E(|S_t|) < \infty$ because the number of seen item groups is finite.
Then we need to prove that $E(|S_{t+1}||S_t,S_{t-1},\cdots,S_0) = E(|S_t|).$ 
First we calculate $E(|S_{t+1}|)$:
\begin{equation}
    E(|S_{t+1}|) = E(|S_{t} \cup Rec_t|).
\end{equation}
\begin{equation}
    E(|S_{t+1}|) = E(|S_{t}|) + |Rec_t| - E(|S_{t} \cap Rec_t|).
\end{equation}
According to Assumption \ref{A3} ,
\begin{equation}
    E(|\{g_j \in S_t^c \cap Rec_t\}|) = 0.
\end{equation}
Therefore 
\begin{equation}
    E(|S_{t} \cap Rec_t|) = |Rec_t|.
\end{equation}
Hence
\begin{equation}
    E(|S_{t+1}|) = E(|S_{t}|).
\end{equation}
By the law of total expectation, we can conclude that
\begin{equation}
    E(|S_{t+1}||S_t,S_{t-1}...S_0) = E(|S_t|).
\end{equation}
Hence the process $|S_t|$ is a martingale defined in the filtered probability space $(\Omega,S,\mathbb{P})$.
\end{proof}
\begin{theorem}
\label{theorem2}
Let $\Delta S_t = |S_{t+1}| - |S_{t}|$ be the martingale difference defined on the filtered space $(\Omega,S,\mathbb{P})$. We define  the average user discovery $\Delta_n S = \frac{1}{n} \sum_{t=0}^{n} \Delta S_t$. Then, with probability $1-\delta$, we have:
\begin{equation}
    \frac{1}{n} \sum_{t=0}^{n} \Delta S_t \leq \frac{\ln(\frac{1}{\delta}) |Rec|^2}{2n}.  
    \end{equation}
\end{theorem}

\begin{proof}
Since $|S_t|$ is a martingale, it is given that $\Delta S_t$ is a martingale difference. Furthermore, we can verify that $\Delta S_t$ is bounded. In fact, $\Delta S_t \geq 0$ because the increase is always positive as iterations advance, and also   $\Delta S_t \leq |Rec_t|$ since the maximum number of new item groups that are going to be added to $|S_t|$ is the number of recommended item groups.
Therefore, using the Azuma-Hoeffding inequality we obtain that $\forall \mu>0$

\begin{equation}
    P(\frac{1}{n} \sum_{t=0}^{n} \Delta S_t > \mu) \leq \exp(-\frac{2 n^2 \mu}{\sum_{t=0}^{n} |Rec_t|^2}),
\end{equation}
and as the number of groups in the recommended list is constant,
\begin{equation}
    P(\frac{1}{n} \sum_{t=0}^{n} \Delta S_t > \mu) \leq \exp(-\frac{2 n \mu}{ |Rec|^2}),
\end{equation}
If we consider $\delta = P(\frac{1}{n} \sum_{t=0}^{n} \Delta S_t > \mu)$, then solving for $\mu$ will give the desired result.
\end{proof}

Lemma \ref{theorem1} and Theorem \ref{theorem2} show that the average discovery of the user is decreasing with the number of recommendation iterations $n$. 
It also relates this decrease to the length of the recommendation list. This is expected because longer recommendation lists increase the chance to see new groups of items.

We conclude our theoretical analysis with two corollaries that provide convergence results for the average \textit{user discovery} and for the \textit{blind spot}. 
\begin{corollary}
\label{corollary1}
The average user discovery $\Delta_n S$ converges to 0 almost surely with increasing iterations. 
\end{corollary}

\begin{proof}
The proof is straightforward from the result of Theorem \ref{theorem2}. In fact, using the Azuma-Hoeffding bound, we can show that  when $n \rightarrow \infty$, $P(\frac{1}{n} \sum_{t=0}^{n} \Delta S_t > \mu ) = 0$, and hence $P(\frac{1}{n} \sum_{t=0}^{n} \Delta S_t = 0 ) = 1$, which concludes the result.
\end{proof}

\textbf{Corollary \ref{corollary1}} further proves that the generic recommender system is rapidly forming a filter bubble \textit{by keeping a blind spot of item groups that is out of reach of the user discovery}. We try to formulate this relationship with Corollary \ref{corollary2}.

\begin{corollary}
\label{corollary2}
Let $|\Delta B_t| = ||B_{t+1}| - |B_{t}||$ be the evolution of the blind spot for a given user. We define  the average decrease of the blind spot as $\Delta_n B = \frac{1}{n} \sum_{t=0}^{n} \Delta B_t$. Given a recommender system with a decreasing error function $e(t)$ (i.e the accuracy of the recommender system increases after each feedback loop iteration), where $e(t) = \left |S_{t}\cap \overline{Rel}  \right | $, then:

If $\Delta_n S$ converges to $0$ almost surely when $n$ tends to infinity, then $\Delta_n B$ converges to $0$ almost surely.

\end{corollary}

\begin{proof}
We have:
\begin{equation}
\label{blind}
    \Delta  B_t= \left |\left | \overline{S_{t+1}} \cap Rel \right |- \left | \overline{S_{t}} \cap Rel \right |\right |.
\end{equation}
Since 
\begin{equation}
    \left | \overline{S_{t+1}} \cap Rel \right | =   |\overline{S_{t+1}}| + |Rel|  - |\overline{S_{t+1}} \cup Rel|
\end{equation}
and 
\begin{equation}
    |\overline{S_{t+1}}| = |G| - |S_{t+1}|,
\end{equation}
we obtain
\begin{equation}
    \label{eq}
    \left | \overline{S_{t+1}} \cap Rel \right | =  |G| - |S_{t+1}| + |Rel| -|\overline{S_{t+1}} \cup Rel|. 
\end{equation}
By inserting (\ref{eq}) in (\ref{blind}) for $t$ and $t+1$, we obtain
\begin{equation}
    \Delta  B_t= \left | - \Delta S_t - \left | \overline{S_{t+1}} \cup Rel \right | + \left | \overline{S_{t}} \cup Rel \right |\right |.
\end{equation}
Hence and considering that $ \left | \overline{S_{t+1}} \cup Rel \right | = |G| - \left | S_{t+1} \cap \overline{Rel} \right |$, we obtain
\begin{equation}
    \Delta  B_t= \left |- \Delta S_t + \left | S_{t+1} \cap \overline{Rel} \right | - \left | S_{t} \cap \overline{Rel} \right |\right |
\end{equation}
\begin{equation}
    \Delta  B_t= \left |- \Delta S_t + e(t+1) - e(t)\right |.
\end{equation}
Therefore using the triangle inequality, we deduce that 
\begin{equation}
    \frac{1}{n}\sum_{t=0}^{n} \Delta B_t \leqslant \Delta_n S + \frac{1}{n} \sum_{t=0}^{n} \left |e(t+1) - e(t) \right |.
\end{equation}
Hence
\begin{equation}
    \frac{1}{n}\sum_{t=0}^{n} \Delta B_t \leqslant \Delta_n S + \frac{1}{n} (e(0)-e(n)).
\end{equation}
Since the error is a decreasing function of $n$,  $\Delta_n B$ is positive, and $\Delta_n S$ converges to 0 almost surely when n tends to infinity, $\Delta_n B$ converges to 0 almost surely when n tends to infinity.
\end{proof}

Our theoretical results prove that a collaborative filtering recommender system that aims at increasing its accuracy by recommending relevant items, is prone to forming a blind spot around a certain group of items (depending on the initial state of the system). This result also shows that the existence of the blind spot is dependent on the strategy of the recommender system. A strategy aiming at improving the relevance of the recommended items will fail under our assumptions at exploring new items.

\section{Empirical evaluation}

\subsection{Ranking Assumption Validation}
\label{Empval}
We start by empirically evaluating the validity of the Ranking \textbf{ Assumption \ref{A3}} by testing the following null hypothesis, where we use the same notation as in\textbf{ Assumption \ref{A3}}:

\noindent\fbox{%
    \parbox{\linewidth}{%
        \textbf{Null Hypothesis ($H_0$):} If two item groups have different distances to the seen group set $S_t$ (one group is already seen and the other group is not seen), then their  ranking probability $P_f$ is not significantly different.
    }%
}

In order to test this hypothesis, we sample a number of users that have not rated certain groups of items and evaluate the ranking behavior of the algorithm toward these groups compared to the seen groups by comparing the averages of the predicted ratings between the two groups. Finally we test the hypothesis by performing a t-test to see if there is a significant difference.

In our experiments, we used the Movielens 1M Dataset \footnote{https://grouplens.org/datasets/movielens/1m/}, which has 6040 users, 3952 items, and one million ratings in total. One reasonable way to group items is based on the genres provided with the dataset. There is a total of 18 genres, some items have many genres, thus we verify that the user has already seen \textbf{any} of the item's genres. The reason behind this grouping choice is to prevent the recommendation algorithm from explicitly using the item groups in the training. To clarify this further, if we grouped the items by their latent representation, as learned by Matrix Factorization \cite{5197422} for instance, our test would be biased due to the fact that this segmentation is used in the prediction step.

For Matrix Factorization, a five-folds cross validation hyperparameter tuning resulted in the following hyperparamters: \{Learning rate: 0.001, Latent dimension: 10, $L_2$ regularization coefficient: 0.01, and number of epochs: 300\}. We used Stochastic Gradient Descent (SGD) to train our model and computed the final rating predictions based on the dot product of the user and item latent factor vectors. We further repeated our experiments 10 times. 

\begin{table}
\caption{Comparison between the mean of the predicted ratings for items from the seen groups (same genre) and items from unseen groups.}
\label{tab:my-table}
\begin{center}
{%
\begin{tabular}{lcll}
\hline
                                  \multicolumn{1}{l}{Item's group} & \multicolumn{1}{l}{Mean of predicted rating} & Variance                                 & P-value                                \\ \hline
\multicolumn{1}{l}{Seen Group }   & \textbf{3.13}                                        & \multicolumn{1}{c}{1.5 $\times 10^{-5}$} & \multirow{2}{*}{2.32 $\times 10^{-10}$} \\ \cline{1-3}
\multicolumn{1}{l}{Unseen Group} & 3                                             & 7 $\times 10^{-4}$                        &                                         \\ \hline
\end{tabular}%
}
\end{center}
\end{table}

\begin{figure}[ht]
\begin{center}
\centerline{\includegraphics[width=0.55\linewidth]{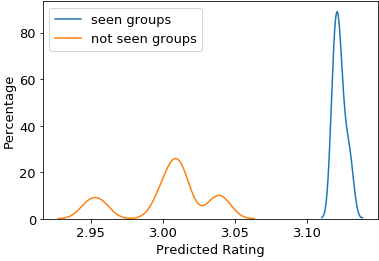}}
\caption{\textbf{The predicted rating distribution of the seen item groups is significantly different from that of the unseen item groups. The predicted ratings affect the ranking of the items in the recommendation list as stated in Assumption \ref{A3}}}
\label{fig:dist}
\end{center}
\end{figure}

The results, shown in \textbf{Table \ref{tab:my-table}}, lead us to reject the null hypothesis since there is a significant difference (p-value $< 0.01$) between the two averages. \textbf{Figure \ref{fig:dist}} further confirms this difference between the two distributions: The seen item groups have a higher predicted rating on average than the items from the unseen group. This confirms, empirically, that our ranking assumption is reasonable. 

What makes this an interesting result is the fact that \textit{we did not use} the genres in the learning and the recommendation step. In fact, the collaborative filtering algorithm was able to detect this dependency based on the similarity of the items' ratings. This is expected as we believe that the recommendation process used to collect the data may have exploited this similarity, and thus exacerbated the limited exposure of the users to different genres.

\subsection{Asymptotic behavior of the user discovery empirical validation}

\subsubsection{Experimental Design}
In order to simulate an iterative recommendation process, we need access to a complete rating matrix. In fact, simulating the sequential dependency needs the true rating of the user to be added at iteration $t-1$ in order to train $f_t$. Unless we perform real world experiments using a production environment, this is impossible to achieve using the available data due to the high sparsity of the data.

For this reason, we use a semi-synthetic data using the Movielens 1M dataset, as was done in \cite{DBLP:journals/corr/SchnabelSSCJ16,khenissi2020modeling}. The process to construct the dataset is done by running a matrix completion algorithm on the original dataset to generate a complete rating matrix. Then, we scale the ratings of each user using a mapping function $g: \hat{R} \rightarrow \{1,2,3,4,5\}$. By defining five percentiles $(p_1,p_2,p_3,p_4,p_5)$ from the predicted ratings of the user, we apply the following mapping: $g(p_1)=1$, $g(p_2)=2$, $g(p_3)=3$, $g(p_4)=4$, $g(p_5)=5$. We remove by this procedure the user specific bias where some users would have a tendency to provide higher ratings or lower ratings compared to others and we scale the predictions into a range that is similar to the original input data. 

The experimental protocol, after constructing the complete matrix, is described in \textbf{Algorithm \ref{alg:exp}}

\begin{algorithm}[tb]
   \caption{Experimental Simulation Steps}
   \label{alg:exp}
\begin{algorithmic}
   \STATE {\bfseries Input:} Ratings $S_t$, Genres, number of runs $nruns$, number of iterations, $iterations$
   \STATE {\bfseries Output:} $(|S_t|)_t , (|B_t|)_t$
   \FOR{$i=1$ {\bfseries to} nruns}
   \FOR{$t=1$ {\bfseries to} iterations}

   \STATE Train $f_t$                    // Train Matrix factorization model.
   \STATE $Rec_t = f(S_t)$              // Select top-n recommendations.
   \STATE $S_{t+1} = S_t \cup Rec_t$.   // Add recommendation's true ratings to  $S_t$
   \STATE Calculate $|S_t|$ and $|B_t|$.
   \ENDFOR
   \ENDFOR

\end{algorithmic}
\end{algorithm}

In our simulations, we first assume that the user rates all the recommended items according to the Perfect Feedback \textbf{Assumption \ref{A2}}. Later, we will explore a relaxation of this assumption where the user sees all the recommended items, but rates a few items based on their ranking in the list, which is similar to \cite{10.1145/2043932.2043955}.
We perform 10 runs and set the length of the recommendation list to $|Rec_t| = 10$.

\subsubsection{Results and discussion}
\begin{figure*}[ht]
\centering    

\subfigure[Evolution of $S_t$ and $B_t$ under perfect feedback]{\includegraphics[ width=0.32 \textwidth]{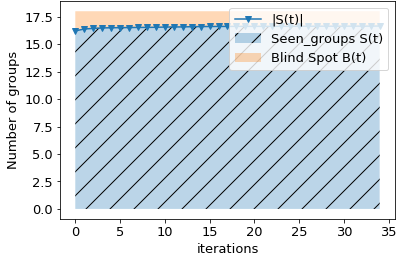} \label{fig:seenbperf}}
\subfigure[Evolution of $ \Delta_n S$under perfect feedback]{\includegraphics[ width=0.32 \textwidth]{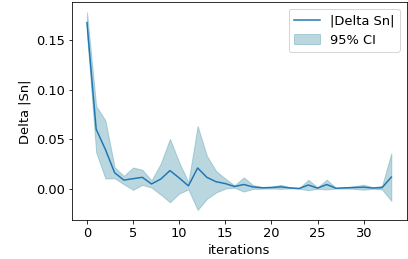} \label{fig:deltasnperf}}
\subfigure[Comparison of $\Delta_n S$ to the the theoretical bound (Theorem \ref{theorem2}) for different levels of $\delta$]{\includegraphics[ width=0.32 \textwidth]{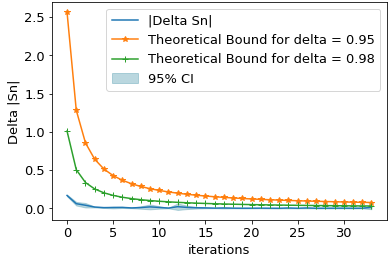}\label{fig:theoperf}
}

\caption{\textbf{Results with the \underline {Perfect Feedback Assumption \ref{A2}}: (a) shows the evolution of $|S_t|$ through iterations. The shaded area shows the cardinality of the seen item groups (genres) and the non-shaded area shows the blind spot's cardinality. (b) shows the evolution $\Delta_n S$ and how it compares to the theoretical bound with different levels of confidence. (b) clearly confirms the hyperbolic decrease of $\Delta_n S$ in Theorem \ref{theorem2}, with the shaded area representing the $95 \%$ confidence interval based on 10 runs.}  }
\label{fig:perf}
\end{figure*}

\begin{figure*}[ht]
\centering    
\subfigure[Evolution of $S_t$ and $B_t$ under imperfect feedback]{\includegraphics[ width=0.32 \textwidth]{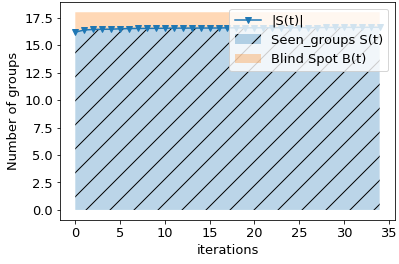} \label{fig:1}}
\subfigure[Evolution of $\Delta_n S$ under imperfect feedback]{\includegraphics[ width=0.32 \textwidth]{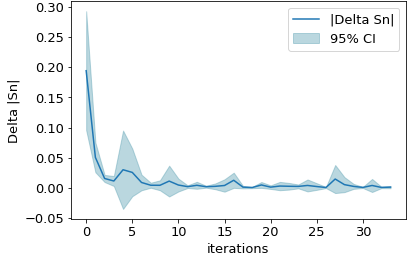} \label{fig:seenbperf}}
\subfigure[Evolution of $S_t$]{\includegraphics[ width=0.32 \textwidth]{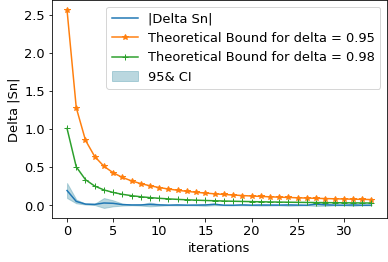} \label{fig:imp}}

\caption{\textbf{Results with the \underline{Relaxation of Perfect Feedback Assumption \ref{A2}}: The user sees the entire recommendation list but only rates a few of them based on a probability that depends on the ranking of the item. We notice that even without the assumption of perfect feedback from the user, the behavior of the algorithm is similar to the previous experience depicted in Figure \ref{fig:perf}}}
\label{fig:imperf}
\end{figure*}

We next investigate the evolution of $|S_t|$ throughout the iterations. \textbf{Figure \ref{fig:seenbperf}} shows how the cardinality of the set of seen item groups stagnates with a negligible increase. Note here that we did not start from the initial state of the system since the available data is collected by running several iterations of recommendations.

We also study the evolution of average user discovery $\Delta_n S$, as shown in \textbf{Figure \ref{fig:deltasnperf}}, which confirms the hyperbolic decrease rate in Theorem \ref{theorem2}. This shows that the average user discovery is shrinking throughout the iterations by converging to zero. The variance present in some of the collected points is due to the fact that our data is semi-synthetic and this affects the feedback loop.
Despite this sporadically high variance, we can see a general decreasing behavior in $\Delta_n S$.

To further validate our theoretical results, we plot the theoretical bound for two different confidence thresholds. \textbf{Figure \ref{fig:theoperf}} shows the asymptotic relationship, with increasing iterations, between the empirical result and the theoretical bound. 

With these results, we showed that despite not including the groups in the learning process of the algorithm, the recommender system fails to recommend new genres to the user and this is exacerbated when the process advances through more iterations. This decrease explains known phenomena such as the filter bubble or echo chambers \cite{Jiang_2019}. These consequences tend to get worse as a result of the feedback loop. 

To investigate the impact of the Perfect Feedback Assumption (\textbf{Assumption \ref{A2}}), we repeated the experiments by making the user rate only a part of the recommendation list according to a discovery probability \cite{10.1145/2043932.2043955} related to an item's rank in the recommendation list (the lower the item is in the list, the smaller the chance it will be rated). The results in \textbf{Figures \ref{fig:imp}} show that there is no significant difference, and the behavior of $|S_t|$ and $\Delta_n S$ remains similar to the previous experiment that was depicted in \textbf{Figure \ref{fig:perf}}. This shows that even without rating all the items and seeing more items than reporting to the algorithm, the algorithm is still failing to discover new groups.

In our experiments, we considered a finite number of groups of items which intuitively should discourage the algorithm from creating a filter bubble (If the number is small, then there is a higher chance for the algorithm to discover all items).
In fact, in practice, the number of items, in particular taking into account new items that get added continuously, may be huge and thus the cardinality of the item space, when considered over time, can be considered to approach the infinite case. Yet, in our experimental process, we have tested on a finite number of groups. The total number of groups in the Movielens 1M data was below 20. Even in this limited scenario,  we could see, as shown in our experimental results, that the algorithm is converging to a state that fails to explore all the groups, again despite the low (finite) number.

\subsection{Effect of a greedy exploration approach on human discovery}
\label{explorationsection}

\begin{figure}[ht]
\centering    

\subfigure[Evolution of $S_t$ and $B_t$ using $\epsilon$-greedy MAB strategy]{\includegraphics[ width=0.48 \textwidth]{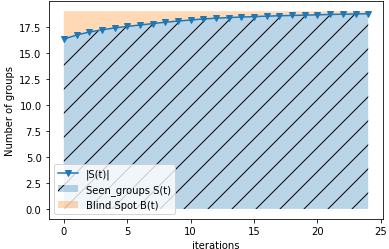} \label{fig:1mab}}
\subfigure[Comparison of $\Delta_n S$ to the theoretical bound using $\epsilon-$greedy MAB strategies for different level of $\epsilon$]{\includegraphics[ width=0.48 \textwidth]{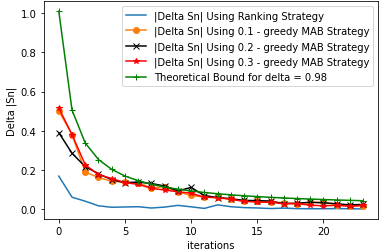} \label{fig:2mab}}

\caption{\textbf{Results using a greedy exploration strategy with Multi-Armed Bandits. \ref{fig:1mab} shows that using a simple random search will decrease the blind spot size, however the increase rate in the human discovery does not exceed the theoretical bound as shown in \ref{fig:2mab}}}
\label{fig:mab}
\end{figure}

To further challenge our assumptions, mainly \textbf{Assumption \ref{A3}}, where we assume a pure exploitation based ranking strategy, we limit the validity of the assumption by using a simple exploration strategy. The main objective is to empirically evaluate the impact of \textbf{Assumption \ref{A3}} on the theoretical findings.

Hence, in this last part of our experiments, we apply a greedy Multi-Armed Bandits (MAB) recommendation strategy \cite{robotica_1999} and provide observations on the behavior of human discovery. We also evaluate the effect of the degree of the exploration by varying $\epsilon$ in the $\epsilon-$greedy MAB approach. Basically, we select a proportion of items in the recommendation list measured by $\epsilon < 1$, where we perform a random selection instead of a ranking approach. This approach violates \textbf{Assumption \ref{A3}} and hence we cannot predict the asymptotic behavior of the human discovery.  

The results in \textbf{Figure \ref{fig:mab}}  shows that Multi-Armed Bandits strategies are effective in reducing the blind spot of the user. However, the average human discovery increase $\Delta_n S$ does not exceed the theoretical bound. This leads us to think that the naive exploration is not enough to provide the user with a discovery potential that is consistent through the iterations. Furthermore, this motivates the search for better exploration methods that target the human discovery explicitly by taking into consideration the theoretical behavior of such phenomena.

These findings, although not surprising, as the bias resulting from recommender systems is becoming well known, should lay a theoretical foundation to more targeted methods to solve this problem. Such solutions may for instance, include advanced loss functions that incorporate $\Delta S$ into the optimization process, new exploration methods to account for the blind spots, or new algorithms that include the user discovery in their modeling.

\section{Conclusion}

We presented a simplified theoretical framework to study the closed feedback loop in a recommender system. We then proved that under a few assumptions, the iterative behavior of the recommender system tends to limit the discovery of the user; and we also proved a theoretical bound on the increase of this discovery. Finally, we presented a simulation framework, using semi-synthetic data, to study the behavior of an iterative recommender system throughout the iterations.

We have also challenged the pure exploitation strategy by experimenting with a Multi-Armed Bandits approach. Results show that our theoretical findings are still valid and thus confirm that the need for more targeted exploration strategies is important.

The limitations of our work come mainly from the assumptions that we have stated. In fact, the user, in real life, is exposed to different external influences, including from different recommender system platforms. Although some of these external exposures may be related either directly (for instance through external data, cookies, etc), or indirectly (through similar social circles), studying an isolated recommender system is limited and needs to be extended to a more general framework.

Our theoretical and empirical  findings do not necessarily translate into a judgement about the acceptability of the convergence of the user discovery and blind spot. It is possible for certain users to be less satisfied with this convergence than others, and the implications vary depending on the item domain (e.g. news vs. movies). That said, we remind the reader that our definition of \textit{blind spot} refers to item groups that are \textit{relevant} to the user.

What we hope to achieve in the future is to combine optimizing for $|\Delta S|$ and $|\Delta B|$ in the training process of the algorithm so that the recommender system can become aware of its own iterative bias after each iteration.

\begin{acks}
This work was supported by National Science Foundation grant NSF-1549981.
\end{acks}

%%
%% The next two lines define the bibliography style to be used, and
%% the bibliography file.
\bibliographystyle{ACM-Reference-Format}
\bibliography{sample-base}

\end{document}